\documentclass[11pt]{article}

\usepackage{arxiv}

\usepackage{amsthm,amsmath,latexsym,bbm,enumitem,xcolor}
\usepackage{amssymb}

\usepackage{graphicx}
\usepackage{mdframed}
\usepackage{physics}

\newtheorem{theorem}{Theorem}

\newtheorem{definition}{Definition}

\newtheorem{property}{Property}

\newcommand{\remove}[1]{}

\title{Detectable Quantum Byzantine Agreement for Any Arbitrary Number of Dishonest Parties}

\author{Vicent Cholvi \\ 
Departament de Llenguatges i Sistemes Inform\`atics \\
Universitat Jaume I, Castell\'o (Spain)}

\begin{document}

\maketitle

\begin{abstract}
Reaching agreement in the presence of arbitrary faults is a fundamental problem in distributed computation, which has been shown  to be unsolvable if one-third of the processes can fail, unless signed messages are used. In this paper, we propose a solution to a variation of the original BA problem, called Detectable Byzantine Agreement (DBA), that does not need to use signed messages. The proposed algorithm uses what we call \emph{$Q$-correlated lists}, which are generated by a \emph{quantum source device}. Once each process has one of these lists, they use them to reach the agreement in a classical manner. Although, in general, the agreement is reached by using $m+1$ rounds (where $m$ is the number of processes that can fail), if less than one-third of the processes fail it only needs one round to reach the agreement.
\end{abstract}

\section{Introduction}

Reaching agreement in the presence of arbitrary faults is a fundamental problem in distributed computation, which has been extensively studied in the past. This problem, also called as Byzantine agreement (BA), consists of several Byzantine generals who are commanding their army divisions to besiege an enemy city. They must decide upon a common plan of action, but they can communicate with one another only by pairwise error-free classical channels. One of the generals, the \emph{commanding general}, must decide on a plan of action and communicate it to the other generals. However, some of the generals, including the \emph{commander} can be dishonest and try to prevent the honest generals from reaching agreement of the plan of action. Thus, the solution to the problem must satisfy:
\begin{description}
\item[IC1:]
All honest parties obey the same order.
\item[IC2:]
If the commanding general is honest, then every honest party obeys the order he sends.
\end{description}

In~\cite{10.1145/322186.322188}, it was shown that this problem is unsolvable if one-third of the generals are dishonest. In~\cite{10.1145/357172.357176}, the authors provided a solution that works for any number of dishonest generals. However, this algorithm and all the subsequent ones that works for any number of dishonest generals, require an authentication structure based on signed messages (e.g.,~\cite{dolev1983authenticated}).

On another hand, in~\cite{PhysRevLett.87.217901} the authors proposed a variation of the original BA problem, called Detectable Byzantine Agreement (DBA), which relaxes the above-mentioned IC1 and IC2 conditions so that all honest parties either perform the same action or all abort. The advantage of using DBA instead of BA is to avoid the use of signed messages. and, as it has been argued in~\cite{PhysRevLett.87.217901}, using DBA is enough for applications where robust tolerance to errors is not necessary and detection suffices.

The authors in~\cite{PhysRevLett.87.217901,PhysRevLett.98.020503,PhysRevLett.100.070504} presented quantum solutions to the DBA problem but for only three parties (the \emph{commander} and two generals).  In~\cite{10.1145/1060590.1060662}, a quantum solution has been proposed that considers any number of parties, but it assumes  that less than one-third of the parties will be dishonest. Another quantum solution that considers any number of parties has been presented in~\cite{IJTP:Quantum}, but also assumes  that less than one-third of the parties will be dishonest. As far as we know, there have been only two proposals to solve the DBA problem for any number of dishonest parties~\cite{Quantum:Takavoli,Sun_2020}, but their agreement solutions are not fully correct (see the Appendix~\ref{sec:counter-examples}).

\paragraph*{Our work}
In this paper, we propose a solution for the DBA problem, without using signed messages, for any number of dishonest generals, which we call \emph{parties}. For this task, we use $Q$-\emph{correlated lists}. Such lists are distributed to the parties by using a number of entangled quantum particles that are generated by a \emph{quantum source device}. Once each party has one of these lists, they use them to reach the agreement in a classical manner.  At this point, our proposed solution has two interesting features:
\begin{enumerate}
\item
On one hand, any forgery of the state of the above mentioned particles (and, therefore, in the $Q$-correlated lists) can be detected. 

\item
On the other hand, the option of abort is considered only in the distribution of the lists. Thus, in the agreement phase, our solution still enables full BA.
\end{enumerate}

%A similar idea has been  used in~\cite{PhysRevLett.87.217901,PhysRevLett.98.020503} to solve the BA problem in the case of only three parties (the \emph{commander} and two generals). In~\cite{10.1145/1060590.1060662}, it has been proposed a quantum solution that considers any number of parties, but assumes  that less than one-third of the parties will be dishonest. Another quantum solution that considers any number of parties has been presented in~\cite{IJTP:Quantum}, but also assumes  that less than one-third of the parties will be dishonest. As far as we know, there have been only two proposals to solve the BA problem for any number of dishonest parties~\cite{Quantum:Takavoli,Sun_2020}, but their agreement solutions are not fully correct (see the Appendix~\ref{sec:counter-examples}).

%\dk{At this point, we note that, although after aborting the distribution of the list the protocol can start again, this feature slightly changes the conditions of the original BA problem. That modification of the original problem, known as detectable Byzantine agreement (DBA)~\cite{PhysRevLett.87.217901}, relaxes the above-mentioned IC1 and IC2 conditions so that all honest parties either perform the same action or all abort (this is enough for applications where robust tolerance to errors is not necessary, and detection suffices).  Anyhow, it must be taken into account that the proposed $QBA(m)$ algorithm still enables full BA, since the option of abort is considered only in the distribution of the $Q$-correlated lists.}

The rest of the paper is structured as follows. In Section~\ref{sec:semi-correlated}, we define $Q$-correlated lists, In Section~\ref{sec:distribution}, we show how the above-mentioned $Q$-correlated lists can be distributed, so that any forgery of their states can be detected. In Section~\ref{sec:QBA-algorithm} we introduce an algorithm that, by using these lists, solves the BA problem in a classical manner without using any quantum resources. We end, in Section~\ref{sec:remarks}, with some open issues.

\section{Sets of $Q$-correlated lists}
\label{sec:semi-correlated}

 In this section, we introduce a data structure, which we call $Q$-\emph{correlated list}, that is the core of the BA algorithm presented in Section~\ref{sec:QBA-algorithm}. In Section~\ref{sec:distribution} we will show how, by using a number of entangled quantum particles, it is possible to provide each party (including the \emph{commander}) with one of the above-mentioned list. 

Given a list $L$, we denote as $L^k$ the element at position $k$ in the list $L$. 
\begin{definition}
Let ${\cal S}=\{L_1, ...,L_n\}$ be a set of $n$ lists, each formed by elements in $W=\{0, 1, \cdots, w \}$, with $w \geq n$. We say that $\cal S$ is $Q$-\emph{correlated} (where $Q$ is a set of positions in the lists) if the following three conditions hold:
\begin{enumerate}
\item
All the lists have the same length.
\item
All the elements are random values in $W$.
\item
For each two different $L_i$ and $L_j$ in ${\cal S} : L_i^k \neq L_j^k$, provided $k \in Q$.
\end{enumerate}
\end{definition}

The positions in $Q$ are called \emph{correlated} positions. Observe that the elements at position $k$ in these lists (i.e., $L_1^k L_2^k \cdots L_{n}^k)$ are either (1) different random numbers in $W$ if $k$ is a correlated position, or (2)  random numbers in $W$ if $k$ is not a correlated position (although these number may be different). Note that,  since the number of elements in $W$ is greater than the number of lists in $\cal S$, from a subset of lists is not possible to infer, with complete certainty, what the others will be, even if it is known which positions are correlated.

\subparagraph*{Example.} Let  ${\cal S}= \{\{1,2,0,0,3,2,3\}, \{2,1,3,0,0,0,2\},\{0,3,1,3,1,1,0\}, \{3,0,2,2,2,3,1\}\}$, with $W=\{0,1,2,3\}$. $\cal S$ is $Q$-correlated with $Q=\{1,2,3,5,6,7\}$, since all the lists have the same length and, at the same correlated positions, the elements take different values. On the contrary, $\cal S$ is not $Q$-correlated with $Q=\{3,4,5\}$, since the fourth element is the same in the first and second lists.\\

\begin{definition}
\label{def:consistent}
Let $v \in W=\{0, 1, \cdots, w \}$ and let ${\cal L}$  be a set of lists each formed by elements in $W$.  We say that the pair $(v,\cal L)$ is \emph{consistent} provided the following three conditions hold:
\begin{enumerate}
\item
All the lists in $\cal L$ have the same length.
\item
All the elements in the lists in $\cal L$ are random values in $W-\{v\}$.
\item
For each two  lists ${\cal L}_i$ and ${\cal L}_j$ in ${\cal L} :  {\cal L}_i^k \neq {\cal L}_j^k$, for all $k$.
\end{enumerate}
\end{definition}

Next, we will state two properties of the $Q$-correlated sets of  lists that will be key in the operation of the proposed agreement algorithm. Given a set of positions $R$, we denote as $L^R$  the list formed by the elements $L^k$ such that $k \in R$, maintaining these elements the same relative order as in $L$. Note that $L^R$ denotes a list of elements, whereas $L^k$ denotes an element.

\begin{property}
\label{property1}
Let $\cal S$ be a $Q$-correlated set of lists, each formed by elements in $W$. Let $v \in W$ and $L_i$ an arbitrary list in $\cal S$. Let $R \subseteq Q$ such that $L_i^k = v$ for all $k \in P$, and ${\cal L}$ a set of lists of the form $L_{j}^R$, where $j \neq i$. The pair $(v,\cal L)$ is consistent,
\end{property}
\begin{proof}
Clearly, all the lists in $\cal L$ have the same length. Since $\cal S$ is $Q$-correlated then the elements at the same positions in the lists in $\cal L$ are different. Furthermore, these values will be different from $v$ (since $v$ appears in $L_i^R$ in all positions). Therefore, the obtained pair will be consistent.
\end{proof}

\subparagraph*{Example.} By using the previous set $\cal S$ with $Q=\{1,2,3,5,6,7\}$, if we know the values of the list $L_1$ then, for $v=2$, we can choose $R=\{2,6\}$ and we  guarantee that any pair $(2,\cal L)$ (with $\cal L$ formed by $L_{j}^R$ lists, where $j \neq 1$) is consistent. 

\begin{property}
\label{property2}
Let $\cal S$ be a $Q$-correlated set of lists, each formed by elements in $W$. Assume that we don't know the values of some arbitrary list $L_i \in \cal S$ and which positions are correlated. Then, it is not possible to choose a set of lists $\cal L$ (not necessarily in $\cal S$), each formed by elements in $W$, and a set  $R$ of positions in these lists, such that the pair $(v, {\cal L}')$ where  ${\cal L}' \equiv \{{\cal L},L_i^R\}$
is guaranteed to be consistent. 
\end{property}
\begin{proof}
Since the number of elements in $W$ is greater than the number of lists in $\cal S$, we cannot identify with complete certainty which are all the correlated positions, even if we know the values of all the lists in $\cal S$, except $L_i$.

Then, assume that we choose $R$ such that it contains a non-correlated position $k$. Since that position is non-correlated, we are not guaranteed that the value at  position $k$ in $L_i^R$ won't be $v$, or any of the values at position $k$ in the lists in $\cal L$, which will make the pair $(v, {\cal L}')$ inconsistent. In other words, we cannot fully guarantee that the pair $(v, {\cal L}')$ will be consistent.
%That will happen w.h.p. as we increase the size of $R$.
\end{proof}

%Observe that Property~\ref{property2} remains true even if the rest of the lists in $\cal S$ are know. Furthermore, that pair will be inconsistent w.h.p. as we increase the size of $R$.

\section{Distributing the $Q$-correlated lists}
\label{sec:distribution}

For the distribution of the $Q$-correlated lists among the parties, we assume that there is a honest independent \emph{quantum source device} (QSD) that will communicate with the parties through pairwise error-free quantum channels. A pairwise quantum channel is said to be \emph{error-free} provided it guarantees that there will be no change in the state of any sent particle due to the own channel, although there is no guarantee that such state could be tampered by third parties. That QSD will prepare and distribute a number of particles so that each party, by  measuring them, will obtain one list $Q$-correlated with the other parties' lists.  

Let $W=\{0, 1, \cdots, w \}$, with $w \geq n$ (where $n$ is the number of parties). The particles that will be distributed are of three types: 

\begin{enumerate}

\item
Particles in the following uniform random states: $\ket{\Psi_0}= \frac{1}{\sqrt{w+1}} \sum_{j=0}^{w} \ket{j}$. Clearly, the measured states of each particle will obtain a random uniform value in $W$.

\item
Particles in the following quantum entangled states: $\ket{\Psi_1}= \frac{1}{\sqrt{w+1}} \sum_{j=0}^{w} \ket{j \otimes j}$. Now, the measured  states of each one single-particle will obtain the same value in $W$.

\item
Particles in the following quantum entangled states:
%These entangles states as quantum resources have been already used in superdense coding~\cite{PhysRevA.65.022304}.

$$\ket{\Psi^s_{i_1,i_2, \cdots, i_{q-1}}}_{q} = \frac{1}{\sqrt{d}} \sum_{j=0}^{d-1} e^\frac{2 \pi i j s}{d}  \ket{j} \otimes  \ket{j + i_1 \mod d} \otimes \cdots \otimes \ket{j + i_{q-1} \mod d} ,$$

where $q$, $i_1, \cdots, i_q \in \{0, 1, \cdots, d-1 \}$. If we take  $s=0$ then we have:

$$\ket{\Psi^0_{i_1,i_2, \cdots, i_{q-1}}}_{q}= \frac{1}{\sqrt{d}} \sum_{j=0}^{d-1}  \ket{j} \otimes  \ket{j + i_1 \mod d} \otimes \cdots \otimes \ket{j + i_{q-1} \mod d}.$$

Let us also we take $q=d=w+1$ and let us perform the measurements of the single-particle states in the base $MB=\{\ket{0}, \ket{1}, \cdots, \ket{w} \}$, denoting the measured state $\ket{0}$ as $0$, $\ket{1}$ as $1$, $\cdots$, $\ket{w}$ as $w$. As it has been shown in~\cite{PhysRevA.65.022304}, if the parameters $i_1, \cdots, i_{w}$ in $\ket{\Psi^0_{i_1,i_2, \cdots, i_{w}}}_{w+1}$ are different then each one of the $w+1$ single-particle measured states will obtain a different value in $W$.

%Let us  assume that the measurement of the single-particle states are performed in the base $MB=\{\ket{0}, \ket{1}, \cdots, \ket{q-1} \}$, and let us denote the single-particle measured state $\ket{0}$ as $0$, $\ket{1}$ as $1$, $\cdots$, $\ket{q-1}$ as $q-1$. Let us take $q=d$ (i.e., the number of particles (which is $q$) will be the same as the dimension $d$) and denote $V=\{0, 1, \cdots, q-1 \}$. Now, when we measure the single-particle states of the entangled state $\ket{\Psi^0_{i_1,i_2, \cdots, i_{q-1}}}_{q}$, we have that if the parameters $i_1, \cdots, i_{q-1}$ are different then each one of the $q$ single-particle measured states will obtain a different value in $V$.

\end{enumerate}

\begin{figure}
\begin{mdframed}
\begin{center}

\begin{enumerate}
\item
Let $W=\{0, 1, \cdots, w \}$, so that $w \geq n$ (where $n$ is the number of parties).

\item
For $t=1$ to $L$, where $L$ denotes the length of the lists, the QSD decides whether position $t$ in the lists will be correlated or not (that decision is taken at random):

\begin{enumerate}

\item
If position $t$ is chosen to be correlated then the QSD prepares $q$ particles in the entangled state $\ket{\Psi^0_{i_1,i_2, \cdots, i_{w}}}_{w+1}$ by taking  parameters with different values.   Then, the QSD  sends one particle to each party except the \emph{commander}, to whom it sends two particles. Furthermore, the QSD also sends a number of decoy correlated particles randomly interspersed with all the others.

\item
If position $t$ is chosen to be non-correlated then the QSD prepares and sends  one particle in the state $\ket{\Psi_0}$ to each party, except the \emph{commander}. In addition, the QSD prepares two particles in the entangled state  $\ket{\Psi_1}$  and sends them to the \emph{commander}. 

\end{enumerate}

\item
The QSD checks the decoy particles. If no tampering is detected, then move to the next step; otherwise, the distribution protocol is aborted.

%eavesdroppers exists in the quantum channels or not. If no eavesdroppers exists, move to the next step; otherwise, the distribution protocol is aborted.

\item
On the reception of the particles, each party (except the \emph{commander}) will measure their state and will generate a list with the obtained values.  

\item
On the reception of the particles, the \emph{commander} will measure their state and use the first particles of each received pair to generate its list. In addition, it will the use the second particles to detect whether a  positions is correlated or not: namely, a position is correlated when the values of each received pair of particles is different.

\end{enumerate}

\end{center}
\end{mdframed}
\caption{The algorithm to distribute the $Q$-correlated lists.}
\label{fig:list_distribution}
\end{figure}

Figure~\ref{fig:list_distribution} shows the full distribution process. Particles of type~1 and~2 will be used to provide uncorrelated values, whereas particles of type~3 will be used to provide correlated ones. This is because particles of type~3 are the only ones that guarantee that, when measured, their values will be different. While the values provided by particles of type~2 will be always the same, the values provided by particles of type~1 may or may not be different; however, if we use a large enough number of particles, we will guarantee with high probability that there will be some case where the values measured by two parties will be equal.

%\db{As it can be readily seen, the set formed by the lists obtained by each party is $Q$-correlated since, by construction, the used quantum particles have been prepared to generate them.}\\

%As it can be readily seen, the used quantum particles have been prepared so that the obtained set of lists will be $Q$-correlated.

Although for the distribution of the $Q$-correlated lists it has been assumed that a honest QSD generates the particles, which are send to the parties through pairwise error-free quantum channels, if anyone obtains information about what the QSD transmits (e.g., the correlated positions or the state of the transmitted particles), such information could be used to generate consistent data and, therefore, to break the subsequent agreement process.  However, we can prevent the particles from being tampered by using a technique similar to the used in~\cite{IJTP:Quantum}, which is based on the well-known \emph{BB84} quantum key distribution protocol~\cite{BEN84}. Next, we succinctly outline how it works (see the referenced article for a more detailed description). First, the QSD generates a number of decoy correlated particles from $\{\ket{0},\ket{1}, \cdots, \ket{w}, F \ket{0}, F \ket{1}, \cdots,  F \ket{w} \}$, where $F$ is the discrete Fourier transform, and randomly insert them into the sent sequences. After the parties receive the sequences, they send the acknowledgements to the QSD, which announces the positions and bases of the decoy particles. Then, the parties measure the particles and return the measurement results to the QSD, which checks whether eavesdroppers exists in the quantum channels or not. That part of the protocol has only two possible outcomes: either use the distributed lists to reach the agreement, or abort.

\section{The $QBA(m)$ algorithm}
\label{sec:QBA-algorithm}

By using the algorithm introduced in the previous section, we can guarantee that each party will have one list of a $Q$-correlated set. Now, in this section we introduce an algorithm that, by using these lists, solves the BA problem in a classical manner without using any quantum resources. 

The code of the above mentioned algorithm, which we called $QBA(m)$, is shown in Fig.~\ref{fig:qalgorithm}. It assumes that the parties can communicate with one another by pairwise \emph{safe} classical channels Namely, we say that a classical channel is safe provided (i) every message that is sent is delivered correctly, (ii) the receiver of a message knows who sent it and (iii) the absence of a message can be detected.  However, since parties (including the \emph{commander}) can be dishonest, they can  send consistent or inconsistent data (see Definition~\ref{def:consistent}). This includes the case where one dishonest party sends consistent data to some parties and inconsistent data (or no data) to the rest. 

As it can be seen in the Step~1 of the algorithm, we require that the distributed lists be of \emph{sufficiently long} length. This requirement is introduced in order to avoid any casually created consistent pair, which can be guaranteed with high probability as we increase the length of the lists. \\

\begin{figure}
\begin{mdframed}
\begin{center}
\begin{enumerate}
\item
Use the algorithm in Figure~\ref{fig:list_distribution} to distribute among the parties a set of $Q$-correlated lists of \emph{sufficiently long} length. As a result, we have that:

\begin{enumerate}
\item
Each party has one list in a $Q$-correlated set of lists.
\item 
The \emph{commander} is the only party that knows which are the correlated positions.
\end{enumerate}

\item
Let $v \in W$ be the order to be transmitted by the \emph{commander} $c$ and let $\cal L=\{\}$. Then, he sends $(P,(v,\cal L))$ to each party $i$ through pairwise error-free classical channels, where $P$ is a list of correlated positions in $L_c$ in which $v$ appears (but not necessarily all the positions).

\item
For each party $i$ (except for the \emph{commander}):

\begin{enumerate}
\item
If it receives  $(P,(v,\cal L))$ from the \emph{commander}:
\begin{enumerate}
\item
Add $L_{i}^{P}$ to $\cal L$.
\item
If $(v,{\cal L})$ is consistent then:
\begin{enumerate}
\item
$V_i=v$
\item
Send $(P,(v,{\cal L}))$ to all the parties.
%Send $(v,match\_values(i,v))$ to the rest of the parties.
\end{enumerate}
\end{enumerate}

\item
For $m+1$ rounds (starting at round $1$), in each round perform: if at round $r$  it receives $(P,(v,\cal L))$:
\begin{enumerate}
\item
Add $L_{i}^{P}$ to $\cal L$.
\item
If $(v,\cal L)$ is consistent, $v \not\in V_i$ and the number of lists in $\cal L$ is $r+1$:
% (note that, at each round, a party can receive several set of pairs):
\begin{enumerate}
\item
Add $v$ to $V_i$.
\item
If $r \leq m$ then send $(P,(v,\cal L))$ to all the parties. 
\end{enumerate}
\end{enumerate}
\item
$V_i$ will be the same for all the honest parties, so they can decide the same.   
\end{enumerate}

\end{enumerate}

\end{center}
\end{mdframed}
\caption{The $QBA(m)$ algorithm for $m$ dishonest parties.}
\label{fig:qalgorithm}

\end{figure}

%As the next theorem shows, the $QBA(m)$ algorithm solves with high probability the Byzantine Agreement problem for $m$ dishonest parties.

\begin{theorem}
The protocol $QBA(m)$ solves with high probability the Byzantine Agreement problem for $m$ dishonest parties.
\end{theorem}
\begin{proof}
\begin{description}
\item[Prove IC2:]
Assume the \emph{commander} is honest. So, every party will receive the same data from the \emph{commander}. Since no dishonest party can forge that data so that it also looks consistent (by Property~\ref{property2} and taking into account that the \emph{commander} is the only one party that knows which positions are correlated),  by  Property~\ref{property1}, the set $V_i$ (for each $i$) will always contain the same and unique value sent by the \emph{commander}. Therefore, all honest parties (at step 3(c)) will  decide the value sent by the \emph{commander}.

\item[Prove IC1:] Assume the \emph{commander} is dishonest. Two honest parties $i$ and $j$ decide the same provided $V_i$ and $V_j$ are the same when they take the decision (i.e., at step 3(c)). Therefore, we only need to prove that if $i$ adds $v$ to $V_i$ then $j$ also adds $v$ to $V_j$. That is, we have to show that $j$ will also receive a consistent tuple with the value $v$.

\begin{enumerate}
\item
If $i$ receives that value at step 3(a) then it sends it to $j$ in step 3(a)iiB, who will add it to $V_j$ (at step 3(b)iiA).
\item
If $i$ adds $v$ to $V_i$ at step 3(b)i then that's because it received at that round consistent data for that value. Now, we have two possibilities:

\begin{itemize}
\item
Party $i$ receives the data before round $m+1$: in this case, $i$ will send that value to $j$ (at step 3(b)iiB), who will add it to $V_j$ (at step 3(b)iiA). 

\item
Party $i$ receives the  data at round $m+1$: in this case, party $i$ won't send any data and, therefore, party $j$  won't receive data with that value. Since there is, at most, $m$ dishonest parties, to consider consistent data at round $m+1$, such a data must contain $m+1$ lists. However, all lists in $\cal L$ different that $L_{x}^{P}$ will make that data inconsistent. Indeed, let's assume that we add a list $L'$ different from $L_{x}^{P}$. Let $v'$ be a value that  appears at position $k$ in list $L'$. We know that, at that position,  there will be different values in the other parties' list (assuming that we know that it is a correlated position; otherwise is even simpler). However, we don't know the concrete values, at that position, in all the other parties' lists (note that $w \geq n$); so,  it could happen that $v'$ appears in another list at the same position, which will certainly happen if $P$ is long enough. Therefore, the addition of $L'$ to $\cal L$ will make the pair inconsistent. Consequently, one of the lists in $\cal L$ (i.e., $L_{x}^{P}$) must be from a honest party, who will have sent consistent data with the value $v$ to all the parties before round $m+1$.  Thus, $v$ will be already included both in $V_i$ and $V_j$.
\end{itemize}
\end{enumerate}
\end{description}

This completes the proof.
\end{proof}

We would like to note that, for the sake of clarity, we have presented  our BA algorithm as simple as possible. However, it can be optimized in some cases. For instance:
\begin{enumerate}
\item
Our algorithm requires $m+1$ rounds to finish, but it can be easily adapted to the case where $m < n/3$, so that the decision is made by using only one round (the approach is similar to that in~\cite{IJTP:Quantum}).

\item
If the absence of  messages can be detected, then it is possible to advance the decision making immediately after detecting that no message has been transmitted at a given round.
\end{enumerate}

\section{Open issues}
\label{sec:remarks}

\begin{enumerate}

\item
Whereas in this paper we assumed that the QSD is an independent device, perhaps the parties themselves could  be used to generate and send the particles. This technique has already been used by Gaertner et al~\cite{PhysRevLett.98.020503} in the case of three parties.

\item
Based on Hardy's correlations~\cite{PhysRevLett.68.2981} and entanglement swapping, the authors in~\cite{PhysRevA.92.042302} have presented a protocol for the original BA problem with three parties.  So, maybe that could also be used to avoid the possibility of abortion, during the distribution process, when considering several parties.
 
\end{enumerate}

\bibliographystyle{plain}
\bibliography{quantum}

\begin{thebibliography}{10}

\bibitem{10.1145/1060590.1060662}
Michael Ben-Or and Avinatan Hassidim.
\newblock Fast quantum byzantine agreement.
\newblock In {\em Proceedings of the Thirty-Seventh Annual ACM Symposium on
  Theory of Computing}, STOC'05, page 481:485, New York, NY, USA, 2005.
  Association for Computing Machinery.

\bibitem{BEN84}
C.~H. Bennett and G.~Brassard.
\newblock {Quantum cryptography: Public key distribution and coin tossing}.
\newblock In {\em Proceedings of IEEE International Conference on Computers,
  Systems, and Signal Processing}, pages 175--179, India, 1984.

\bibitem{dolev1983authenticated}
Danny Dolev and H.~Raymond Strong.
\newblock Authenticated algorithms for byzantine agreement.
\newblock {\em SIAM Journal on Computing}, 12(4):656--666, 1983.

\bibitem{PhysRevLett.87.217901}
Matthias Fitzi, Nicolas Gisin, and Ueli Maurer.
\newblock Quantum solution to the byzantine agreement problem.
\newblock {\em Phys. Rev. Lett.}, 87:217901, Nov 2001.

\bibitem{PhysRevLett.98.020503}
S.~Gaertner, C.~Kurtsiefer, M.~Bourennane, and H.~Weinfurter.
\newblock Experimental demonstration of four-party quantum secret sharing.
\newblock {\em Phys. Rev. Lett.}, 98:020503, Jan 2007.

\bibitem{PhysRevLett.100.070504}
Sascha Gaertner, Mohamed Bourennane, Christian Kurtsiefer, Ad\'an Cabello, and
  Harald Weinfurter.
\newblock Experimental demonstration of a quantum protocol for byzantine
  agreement and liar detection.
\newblock {\em Phys. Rev. Lett.}, 100:070504, Feb 2008.

\bibitem{PhysRevLett.68.2981}
Lucien Hardy.
\newblock Quantum mechanics, local realistic theories, and lorentz-invariant
  realistic theories.
\newblock {\em Phys. Rev. Lett.}, 68:2981--2984, May 1992.

\bibitem{10.1145/357172.357176}
Leslie Lamport, Robert Shostak, and Marshall Pease.
\newblock The byzantine generals problem.
\newblock {\em ACM Trans. Program. Lang. Syst.}, 4(3):382?401, July 1982.

\bibitem{PhysRevA.65.022304}
X.~S. Liu, G.~L. Long, D.~M. Tong, and Feng Li.
\newblock General scheme for superdense coding between multiparties.
\newblock {\em Phys. Rev. A}, 65:022304, Jan 2002.

\bibitem{IJTP:Quantum}
Qing-bin Luo, Kai-yuan Feng, and Ming-hui Zheng.
\newblock Quantum multi-valued byzantine agreement based on d-dimensional
  entangled states.
\newblock {\em International Journal of Theoretical Physics},
  58(12):4025--4032, 2019.

\bibitem{10.1145/322186.322188}
M.~Pease, R.~Shostak, and L.~Lamport.
\newblock Reaching agreement in the presence of faults.
\newblock {\em J. ACM}, 27(2):228--234, April 1980.

\bibitem{PhysRevA.92.042302}
Ramij Rahaman, Marcin Wie\ifmmode~\acute{s}\else \'{s}\fi{}niak, and Marek
  \ifmmode~\dot{Z}\else \.{Z}\fi{}ukowski.
\newblock Quantum byzantine agreement via hardy correlations and entanglement
  swapping.
\newblock {\em Phys. Rev. A}, 92:042302, Oct 2015.

\bibitem{Sun_2020}
Xin Sun, Piotr Kulicki, and Mirek Sopek.
\newblock Multi-party quantum byzantine agreement without entanglement.
\newblock {\em Entropy}, 22(10):1152, Oct 2020.

\bibitem{Quantum:Takavoli}
Armin Tavakoli, Ad\'an Cabello, Marek Zukowski, and Mohamed Bourennane.
\newblock Quantum clock synchronization with a single qudit.
\newblock {\em Scientific reports}, 5:7982, 01 2015.

\end{thebibliography}

\appendix
\section{Counter-examples}
\label{sec:counter-examples}

\begin{itemize}
\item
Takavoli et al.~\cite{Quantum:Takavoli}: This algorithm is intended to solve binary DBA. In the algorithm in Table~1, assume $P_1$ is faulty and sends consistent pairs to all the processes, so that all messages are the same, except one. Now assume that the process that receives the different message (which is also faulty) conveys its received pair to some processes, and $\perp$ to the rest: the processes that receive the pair will decide to abort (since they detect, by (iib), that $P_1$ is faulty), but those who receive $\perp$ will decide the value sent by $P_1$ (they apply (iid)). That is, non-faulty processes will decide different things.

Furthermore, the quantum protocol used for distributing the correlated lists has not been shown to be always correct. For instance, it could happen that a dishonest process reveals a fake encoding base (e.g., choosing it at random) so that, by chance, the sum of the basis choices modulo $m$ equals zero, while the sum of the right basis choices modulo $m$ is different from zero. In that case, the run would be treated as a valid distribution of the numbers at the same position in the private lists. That is enough to break the subsequent Byzantine agreement algorithm.

\item
Sun et al.~\cite{Sun_2020}: This algorithm is intended to solve multivalued DBA. At stage~2, assume that $P_1$ is faulty and sends consistent pairs to all the processes, so that all messages are the same, except one. Now, assume that the process that receives the different value (which is also faulty) conveys its received pair to some processes, and $\perp$ to the rest: the processes that receive the consistent pair will decide $\perp$ (they will apply 3(a)), but those who receive $\perp$ will decide the value send by $P_1$ (they will apply 3(c)). That is, non-faulty processes will decide different things.

\end{itemize}

%\section{Consensus}
%\label{sec:consensus}
%
%\begin{enumerate}
%\item
%Each party acts as a commander and transmits one value to the rest of parties. \db{The QSD sends the Q-correlated particles for each party.}
%
%\item
%The honest parties decide for each one value as in the Byzantine agreement problem.
%\item
%The honest parties finish with the same values: they decide the same.
%\end{enumerate}
%
%\section{Open questions}
%
%\dk{Is it possible that the parties send the Q-correlated lists?. For honest parties, that's not a problem.}
%
%\dk{Is it possible for honest parties to detect that a list is fake? Even if they abort the decision.}

\end{document}